\documentclass[12pt]{article}
\usepackage{times} 
\usepackage{algorithm}
\usepackage{algpseudocode}
\usepackage{amsmath,amssymb,amsfonts,mathrsfs}
\usepackage{a4wide}
\usepackage{url}
\usepackage{hyperref}
\usepackage[english]{babel} 
\usepackage{amsthm}
\usepackage{xcolor}

\usepackage{natbib}
\newtheorem{theorem}{Theorem}
\newtheorem{definition}{Definition}
\newtheorem{corollary}{Corollary}
\newtheorem{proposition}[theorem]{Proposition}{\bfseries}{\itshape}
\newtheorem{lemma}{Lemma}
{\bfseries}{\itshape}
\newtheorem{notation}{Notation}

\renewcommand{\leq}{\leqslant}
\renewcommand{\le}{\leqslant}
\renewcommand{\geq}{\geqslant}
\renewcommand{\ge}{\geqslant}

\newcommand{\eqdef}{\stackrel{\text{def}}{=}}
\newcommand{\F}{\mathbb{F}}
\newcommand{\K}{\mathbb{K}}
\newcommand{\fq}{\F_{q}}
\newcommand{\fqm}{\F_{q^m}}

\newcommand{\word}[1]{\boldsymbol{#1}}
\newcommand{\bv}{\word{b}}
\newcommand{\cv}{\word{c}}
\newcommand{\ev}{\word{e}}
\newcommand{\gv}{\word{g}}

\newcommand{\mv}{\word{m}}

\newcommand{\xv}{\word{x}}

\newcommand{\zv}{\word{z}}

\newcommand{\mat}[1]{\boldsymbol{#1}}
\newcommand{\Am}{\mat{A}}
\newcommand{\Bm}{\mat{B}}

\newcommand{\Gm}{\mat{G}}
\newcommand{\Jm}{\mat{J}}
\renewcommand{\Im}{\mat{I}}

\newcommand{\Mm}{\mat{M}}
\newcommand{\Hm}{\mat{H}}
\newcommand{\Pm}{\mat{P}}
\newcommand{\Qm}{\mat{Q}}
\newcommand{\Rm}{\mat{R}}
\newcommand{\Sm}{\mat{S}}
\newcommand{\Tm}{\mat{T}}
\newcommand{\Xm}{\mat{X}}

\newcommand{\ZZ}{\mat{0}}
\newcommand{\Gp}{\mat{G}_{\rm pub}}
\newcommand{\Cp}{\code{C}_{\rm pub}}

\newcommand{\pk}{\mathsf{pk}}
\newcommand{\sk}{\mathsf{sk}}

\newcommand{\KG}{\mathsf{KeyGen}}

\newcommand{\Enc}{\mathsf{Encrypt}}
\newcommand{\Dec}{\mathsf{Decrypt}}

\newcommand{\gab}[2]{\CG_{#1}\left(#2\right)}
\newcommand{\code}[1]{\mathscr{#1}}
\newcommand{\dual}[1]{{#1}^\bot}
\newcommand{\CA}{\code{A}}
\newcommand{\CC}{\code{C}}
\newcommand{\CG}{\code{G}}

\newcommand{\norm}[1]{\left | #1  \right |}
\newcommand{\rank}{{\normalfont \texttt{rank}}}
\newcommand{\GL}{{\normalfont \textsf{GL}}}

\newcommand{\MS}[3]{\mathcal{M}_{#1,#2}\left(#3\right)}

\begin{document}
\title{On the Failure of the Smart Approach of the GPT Cryptosystem}
\date{}
\author{Herv\'e Tal\'e Kalachi } 
\newcommand{\Addresses}{{
  \bigskip
  \footnotesize

  Herv\'e Tal\'e Kalachi, \textsc{Department of Computer Engineering, National Advanced school of Engineering of Yaounde, University of Yaounde 1,
    Po.Box: 8069 Yaounde-Cameroon}\par\nopagebreak
  \textit{E-mail address}: \texttt{herve.tale@univ-yaounde1.cm}
}}

\maketitle

\begin{abstract}
	This paper describes a new algorithm for breaking the smart approach \citep{RGH10} of the  GPT cryptosystem \citep{GPT91}. We show that by puncturing the public code several times on specific positions, we get a public code on which  applying  the Frobenius operator appropriately allows to build an alternative secret key.

\end{abstract}

\section{Introduction}

Code-based cryptography first appeared in 1978, when \cite{M78} proposed the first public key encryption scheme not based on number theory primitives. More precisely, he built a scheme (with Goppa codes) for which the security stands on two problems, namely the hardness of the Syndrome Decoding Problem in Hamming metric \citep{BMT78} and the difficulty to distinguish a binary Goppa code from a random linear code \citep{CFS01,S02}. The public key is formed with a matrix $\Gp$ which is a product of three matrices $\Sm$, $\Gm$ and $\Pm$ where $\Gm$ is a generator matrix of the Goppa code, $\Pm$ a permutation matrix and $\Sm$ an invertible matrix.  The scheme has various advantages such as :
\begin{itemize}
\item the encryption and decryption algorithms are very efficient compared to RSA.
\item the syndrome decoding problem is known to be NP complete \citep{BMT78}, and the best attacks for solving it are exponentials in the code length. 
\item the scheme is a potential candidate for post-quantum cryptography : classic McEliece scheme and some related schemes such as ROLLO and RQC are among the second Round Candidates of the NIST competition for post-quantum cryptography. 
\end{itemize} 

However McEliece scheme came with a big disadvantage: the size of the public keys is about five hundred thousands bits for a security level of only 80 bits \citep{BLP08}. In order to solve this problem, several modifications of the scheme have been proposed among which the use of rank metric codes \citep{GPT91} instead of the Hamming metric. 

The first rank-metric scheme was proposed by \cite{GPT91} and is now called the GPT cryptosystem.
This scheme  can be seen as an analog of the McEliece public key cryptosystem based on the class of Gabidulin codes, with the difference that the new public matrix if given by $\Gp = \Sm(\Gm + \Xm)$, where $\Xm$ is a random matrix with a prescribed rank $t_{\Xm}$ (the matrix $\Xm$ is usually called the distortion matrix).

Gabidulin codes are often seen as ``analogs'' of Reed-Solomon codes \citep{N86,SS92} in the Hamming metric and like them, 
they are highly structured. That is the reason why their use in the GPT cryptosystem has been the subject of
several attacks. Gibson was the first to prove
the weakness of the system through a series of successful attacks \citep{G95,G96}.
Following this failures, the first works which modified the GPT scheme to avoid Gibson's attack were published by \cite{GO01} followed by \cite{GOHA03}.  
The idea is to hide further the structure of Gabidulin code by considering  isometries for the rank metric.
Consequently, a \emph{right column scrambler} $\Pm$ is  introduced. This matrix $\Pm$ is an invertible matrix with its entries in the base field $\fq$,
while the ambient space of the Gabidulin code is $\fqm^n$. 
But \cite{O05a,O05,O08} designed a more
general attack that dismantled all the existing  modified GPT cryptosystems. 
His approach consists in applying an operator $\Lambda_i$, which applies $i$ times the Frobenius operation 
on the public generator matrix $\Gp$. Overbeck observed that the dimension increases by $1$ each time the Frobenius is applied. He then
proved that by taking $i = n - k - 1$, the codimension becomes $1$ if $k$ is the rank of $\Gp$ (which is also the dimension of the associated Gabidulin code).
This phenomenon is a clearly distinguishing property of a Gabidulin code which cannot be encountered for instance with a random linear code where the dimension would increase by $k$ for each use of the Frobenius operator.

Overbeck's attack uses crucially  two important facts, namely the column scrambler matrix $\Pm$ is defined on the 
based field $\fq$ and the codimension of $\Lambda _{n-k-1}\left(\Gp \right)$ is equal to $1$.
Several works then proposed to resist to this attack  either by taking special 
distortion matrix so that the second property is not true as proposed by \cite{L10} and \cite{RGH10}, or 
by taking a column scrambler matrix defined over the extension field $\fqm$ as in \cite{G08,GRH09,RGH11,GP13,GP14}. It was recently shown by \cite{OTN18} that the general construction from \cite{G08} is vulnerable to a structural attack. This attack implies an attack on the variants from \cite{GRH09,RGH11,GP14,GP13}, since the construction of \cite{G08} is a generalization of the constructions given by \cite{GRH09,RGH11,GP14,GP13}.   
\cite{HMR17} also presented a cryptanalysis of the variant of \cite{RGH10}. This attack focus instead on recovering elements of rank one, rather than the dimension of the dual space.
All the above attacks are structural attacks, which unlike generic attacks (or decoding attacks) exploit the structure of the public matrix. We refer the reader to the article from \cite{BBCGPSTV20} for recent progress in term of generic attacks in rank metric.  

In this paper, we present a new structural attack of the variant described by \cite{RGH10}. Contrary to the result of \cite{HMR17}, our attack exploits the structure of the dual space and is based on a new view of the ``smart'' approach \citep{RGH10}. Concretely, we show that the reparation of \cite{RGH10} is equivalent to insert some redundancies in the public code of a standard GPT cryptosystem. We then show how to remove the redundancies in order to be able to apply Overbeck's attack on the public code obtained.


\section{Preliminary Facts}\label{Preliminary}

The finite field with $q$ elements is denoted by $\F_q$ where $q$ is a power of a prime number.
For any subfield $\K \subseteq \F$ of a field $\F$
and for any positive integers  $k$ and $n$ such that $k \le n$,
the $\K$-vector space spanned by  $\bv_1,\dots{},\bv_k$  where each $\bv_i \in \F^n$ is denoted by 
$\sum_{i=1}^k \K \; \bv_i$.
The set of matrices with $m$ rows and $n$ columns 
and entries in $\F$ is denoted by $\MS{m}{n}{\F}$. 
The group of invertible matrices of size $n$ over $\F$ is denoted by $\GL_n(\F)$.

\begin{definition}[Rank weight]
Let $\Am$ be a matrix from $\MS{m}{n}{\F}$ where $m$ and $n$ are positive integers. The \emph{rank weight} 
of $\Am$ denoted by $\norm{\Am}$ is the rank of $\Am$. The \emph{rank distance} between two matrices 
$\Am$ and $\Bm$ from $\MS{m}{n}{\F}$ is defined as $\norm{\Am - \Bm}$.
\end{definition}

It is a well-known fact that the rank distance on $\MS{m}{n}{\F}$ has the properties of a metric. 
In the context of the rank-metric cryptography, this rank distance is rather defined for vectors $\xv \in \fqm^n$ \citep{Gab85}. 
The idea is to consider the field $\fqm$ as an
$\fq$-vector space and hence any vector  $\xv \in \fqm^n$ as a matrix from $\MS{m}{n}{\fq}$ by decomposing 
each entry $x_i \in \fqm$ into an $m$-tuple of $\fq^m$ with respect to an arbitrary basis of $\fqm$.  The rank weight 
of $\xv$ also denoted by $\norm{\xv}$ is then its rank\footnote{This rank is of course independent of the choice of the basis  of $\fqm$ since the rank of a matrix is invariant when multiplied by an invertible matrix.} 
viewed as a matrix of $\MS{m}{n}{\fq}$. Hence, it is possible to define a new metric on $\fqm^n$ that we recall explicitly in the following.

\begin{definition}
Let us consider the finite field extension $\fqm/\fq$ of degree $m \ge 1$.
The \emph{rank weight} of a vector $\xv = \left(x_{1},x_{2},...,x_{n}\right)$ in $\fqm^n$  denoted by $\norm \xv$
is  the dimension of the $\fq$-vector space generated by $\{x_1,\dots{},x_n\}$
\begin{equation}
\norm \xv = \dim \sum_{i=1}^n \fq x_i.
\end{equation}
Similarly,  the  column rank over $\fq$ of a matrix $\Mm$ from $\MS{k}{n}{\fqm}$ is denoted by $\norm \Mm$, and is defined to be the dimension of $\sum_{i}^n \fq \Mm_{i}$ where $\Mm_1,\dots{},\Mm_n$ 
are the columns of $\Mm$.   
\end{definition}
\begin{proposition} \label{prop:rank_reduction}
Let $\Mm$ be a matrix from $\MS{k}{n}{\fqm}$ and set  $s = \norm{\Mm}$ with $s  < n$. 
There exist then $\Mm^*$ in $\MS{k}{s}{\fqm}$ with $\norm{\Mm^*} = s$ and
$\Tm$ in $\GL_n(\fq)$ such that:
\begin{equation}
\Mm \Tm= (\Mm^* \mid \ZZ)
\end{equation}
In particular for any $\xv \in \fqm^{n}$ such that  $\norm{\xv} = s$ 
there exists  $\Tm$ in $\GL_n(\fq)$  for which $\xv \Tm=(\xv^* \mid \ZZ)$ where 
$\xv^* \in \fqm^{s}$ and $\norm{\xv^*}=s$.
\end{proposition}
This permits to state the following corollary.

\begin{corollary} \label{cor:uprank}
For any $\Mm \in \MS{k}{n}{\fqm}$ and for any  $\mv \in \fqm^k$
\begin{equation}
\norm{\mv\Mm} \leq \norm{\Mm}
\end{equation}
\end{corollary}
\begin{proof}
Suppose that $\norm \Mm = s$ and let $\Tm$ in $\GL_n(\fq)$ such that $\Mm \Tm= (\Mm^* \mid \ZZ)$ with $\Mm^*$ in $\MS{k}{s}{\fqm}$ . We then have $\norm{\mv\Mm} = \norm{\mv\Mm \Tm}= \norm{\mv(\Mm^* \mid \ZZ)} \leq \norm{\Mm^*} \leq s$. 
\end{proof}

\begin{notation}
For any $x$ in $\fqm$ and for any integer $i$, the quantity $x^{q^{i}}$ is denoted by $x^{[i]}$. This notation is extended to vectors  $\xv^{[i]} = (x_{1}^{[i]},\dots{},x_{n}^{[i]})$ and matrices $\Mm^{[i]}= \left ( m_{ij}^{[i]} \right)$.
\end{notation}

We also give a lemma that will be useful in the sequel.

\begin{lemma} \label{lem:frob_prop}
For any $\Am \in \MS{\ell}{s}{\fqm}$ and  $\Bm \in \MS{k}{n}{\fqm}$, and for any $\alpha$ and $\beta$ in $\fq$:
\begin{enumerate}
\item If $\ell = k$ and $s= n$ then
\[
\left( \alpha \Am + \beta \Bm \right)^{\left[i\right]} = \alpha \Am^{\left[i\right]} + \beta \Bm^{\left[i\right]}
\]
\item If $s=k$ then
\[
\left(\Am \Bm \right)^{\left[i\right]}=\Am^{\left[i\right]} \Bm^{\left[i\right]}.
\]
In particular if $\Sm$ is in $\GL_n(\fqm)$ then $\Sm^{\left[i\right]}$ also belongs to $\GL_n(\fqm)$ and 
\[
(\Sm^{[i]})^{-1} = \left(\Sm^{-1}\right)^{[i]}
\] 
\end{enumerate}
\end{lemma}
\begin{proof}
The proof of the two points comes directly from the properties of the Frobenius operators (multiplicative and $\fq-$linear). To finish, remark that for $\Sm$ in $\GL_n(\fqm)$, since $\Sm \Sm^{-1}= \Im_n$ we also have $\Sm^{[i]} \left(\Sm^{-1}\right)^{[i]}= \Im_n$. This implies that $\Sm^{[i]}$ belongs to $\GL_n(\fqm)$ and $(\Sm^{[i]})^{-1} = \left(\Sm^{-1}\right)^{[i]}$  
\end{proof}

Let us recall that a (linear) code of length $n$ over a finite field $\F$ is a linear subspace of $\F^n$. Elements of a code are called \emph{codeword}. A matrix whose rows form a basis of a code is called a \emph{generator matrix}.
The dual of a code $\CC\subset \F^n$ is the linear space denoted by $\dual{\CC}$ containing vectors $\zv \in \F^n$ such that:
\[
\forall \cv \in \CC, \;\; \sum_{i=1}^n c_i z_i = 0. 
\]
An algorithm $D$ is said to 
decode $t$ errors in a code $\CC \subset \F^n$ if for any $\cv \in \CC$ and for any $\ev \in \F^n$ such that $\norm \ev \le t$ we have $D(\cv + \ev) = \cv$. 
Generally, we call such a vector $\ev$ an \emph{error} vector.
We introduce now an important family of codes known for having an efficient decoding algorithm.
\subsection{Gabidulin Codes}
\begin{definition}\citep{Gab85}
Let $\gv \in \fqm^{n}$ such that $\norm{\gv}=n$.
The $(n,k)-$Gabidulin code $\gab{k}{\gv}$ is the code of length $n$ and dimension $k$ generated by the matrix 
\begin{equation} \label{gab:genmat}
\Gm=
\begin{pmatrix}
g_{1}^{[0]} & \cdots{} & g_{n}^{[0]} \\
\vdots{}      &              &  \vdots{} \\
g_{1}^{[k-1]} & \cdots{} & g_{n}^{[k-1]}
\end{pmatrix}.
\end{equation}
\end{definition}
Gabidulin codes are known to possess a fast decoding algorithm that 
can decode errors of weight $t$ provided that $t \leq \lfloor \frac{1}{2}(n-k) \rfloor$.
Furthermore the dual of a Gabidulin code $\gab{k}{\gv}$ is also a Gabidulin code. We refer the reader to the papers from \cite{Gab85,L07,SK09,W13} for more details concerning Gabidulin codes and their decoding algorithms.

In the sequel, a matrix that has the structure of \eqref{gab:genmat} will be called a $q-$Vandermonde matrix.  The following proposition gives an important well-known property about Gabidulin codes.

\begin{proposition}\label{prop:ScramblingAGabidulinAtRight}
Let $\gab{k}{\gv}$ be a Gabidulin code of length $n$ with generator
matrix $\Gm$ and $\Tm \in \GL_n(\fq)$. 
Then $\Gm\Tm$ is a generator matrix of the Gabidulin code $\gab{k}{\gv \Tm}.$
\end{proposition}

\begin{proof}

From Lemma~\ref{lem:frob_prop}, we have $\left(\gv \Tm \right)^{\left[ i \right]}=\gv^{{\left[ i \right]}}\Tm$ for any integer $i$. 
\end{proof}
In the following section, we briefly describe the general GPT cryptosystem \citep{GPT91}.
\subsection{The general GPT cryptosystem}
The key generation algorithm of the general GPT cryptosystem takes as input the integers $k$, $\ell$, $n$ and $m$ such that $k<n \leq m$ and $\ell \ll n$ and outputs the public key/private key pair $(\pk,\sk).$

\paragraph{$\KG(n,m,k,\ell,q)=(\pk,\sk)$}
\begin{enumerate}
\item Let $\Gm \in \MS{k}{n}{\fqm}$ be a generator matrix of the Gabidulin code $\gab{k}{\gv}$

\item Pick $\Sm \in \GL_k(\fqm)$, $\Xm \in \MS{k}{\ell}{\fqm}$ and $\Pm \in \GL_{n+\ell}(\fq)$. 

\item Compute $\Gp \eqdef \Sm (\Xm \mid \Gm) \Pm$ and $t=\frac{n-k}{2}$

\item Return 
\[
\pk=(\Gp,t) \text{~and~} \sk=(\Sm, \Pm ).
\] 

\end{enumerate}

To encrypt a message $\mv \in \fqm^{k}$, apply the following function:

\paragraph{$\Enc(\mv,\pk)=\zv$}
\begin{enumerate}
\item Generate a random error-vector $\ev \in \fqm^n$ with $ \norm{ \ev}_q  \leq t$
\item Return $\zv=\mv {\Gp} \oplus \ev$
\end{enumerate}

The decryption function $\Dec( )$ takes as input a ciphertext $\zv$ and the private key $\sk$ and outputs the corresponding message $\mv$.
\paragraph{$\Dec(\zv,\sk) = \mv^\prime$}
\begin{enumerate} 
\item First compute $\zv \Pm^{-1}=\mv \Sm \left(\Xm \mid \Gm \right)+\ev \Pm^{-1}$. 
\item Extract the last $n$ components $\zv^\prime$ of $\zv \Pm^{-1}$. 
\item Apply a fast decoding algorithm of $\gab{k}{\gv}$ to $\zv^\prime$ to obtain $\mv^*$
\item Return $\mv^* \Sm^{-1}$
\end{enumerate}
Remark that $\zv^\prime$ will satisfy  $\zv^\prime =\mv \Sm\Gm +\ev^\prime$ where $\ev^\prime$ is a sub-vector of $\ev \Pm^{-1}$, hence $\norm{\ev^\prime}_q \leq t$. The output $\mv^*$ of the decoding algorithm for $\gab{k}{\gv}$ then satisfies $\mv^* = \mv \Sm$. So we have $\mv^\prime = \mv$.  
\cite{O08} proposed a very efficient attack on the general GPT cryptosystem. The attack exploits the distinguishing properties of Gabidulin codes.

\color{black}
\subsection{Distinguishing Properties of Gabidulin Codes}

We recall important algebraic properties about Gabidulin codes. One key property is that Gabidulin codes can be easily distinguished from random linear codes.

\begin{definition}
For any integer $i \ge 0$ let
$\Lambda _{i} :  \MS{k}{n}{\fqm}  \longrightarrow \MS{(i+1) k }{n}{\fqm}$ be the $\fq$-linear operator that
maps any $\Mm$ from $\MS{k}{n}{\fqm}$ to $\Lambda_i(\Mm)$ where by definition:
\begin{equation}
\Lambda_i(\Mm) \eqdef 
\begin{pmatrix}
\Mm^{[0]} \\
\vdots{} \\
\Mm^{[i]}
\end{pmatrix}.
\end{equation}
For any code $\code{G}$ generated by a matrix $\Gm$ we denote by $\Lambda_i(\code{G})$ the code generated by  $\Lambda_i(\Gm)$.
\end{definition}

The importance of $\Lambda _{i}$ becomes clear when one compares the dimension
of the code spanned by $\Lambda _{i}(\Gm)$ for a randomly drawn matrix $\Gm$ and  the dimension obtained when
$\Gm$ generates a Gabidulin code.

\begin{proposition}\citep{O08} \label{prop:dsg_gab}
Let $\gv$  be in $\fqm^n$ with $\norm{\gv} = n$ with $n \leq m$. For any integers $k$ and $i$ such that $k \leq n$ and
$i\leq n-k-1$ we have:
\begin{equation}
\Lambda _{i}\big(\gab{k}{\gv}\big) = \gab{k+i}{\gv}.
\end{equation}
\end{proposition}

\begin{proposition}\citep{O08}
If $\CA \subset \fqm^n$ is a code of dimension $k$ generated by a random matrix from $\MS{k}{n}{\fqm}$ then, with  probability greater than $1 - 4 q^{-m}$ we have:
\begin{equation}
\dim \Lambda _{i}(\CA) = 
\min\big \{ n,(i+1)k\big\}
\end{equation}
\end{proposition}

In the case of a Gabidulin code, we get a different  situation as explained by Proposition~\ref{prop:dsg_gab}, namely $\dim \Lambda _{i}\big(\gab{k}{\gv}\big) = \min\big \{ n,k + i \big\} $.
Therefore, we have a property that is verifiable in polynomial time, and allowing to distinguish a Gabidulin code from a random code. This can be used in a cryptanalysis context. In fact, \cite{O08} 
has proven that, for a public matrix $\Gp$ given by $\Gp = \Sm \left(\Xm \mid \Gm \right) \Pm$
with $\Xm \in \MS{k}{\ell}{\fqm}$, $\Pm \in \GL_{n+\ell}(\fq)$ and $\Gm$ generating a Gabdidulin code $\gab{k}{\gv}$,
(in particular all the entries of $\Pm$ belong to $\fq$), 
it is possible (under certain conditions) to find in polynomial time an alternative decomposition of 
$\Gm_p$ of the from $\Sm^* \left( \Xm^* \mid \Gm^* \right) \Pm^* $ using the operator $\Lambda_i$. 
This decomposition allows to decrypt any ciphertext computed with $\Gm_p$. 
The reader can refer to \citep{O08} for details concerning the attack.  
The key reason explaining its success is given by the following proposition.

\begin{proposition}
Let us consider $\ell$, $k$ and $n$ be positive integers with $\ell < n$ and $1 \leq k < n$. Let
$\Gm$ be in $\MS{k}{n}{\fqm}$ as a generator matrix of a Gabidulin code, and $\Xm$ 
be a randomly drawn matrix from $\MS{k}{\ell}{\fqm}$. Denote $\CA$ as the
code defined by the generator matrix $\left( \Xm \mid \Gm\right) $. 
Then for any integer $i \ge 0$
\begin{equation} \label{ineq:dsg}
k+i ~ \leq ~ \dim \Lambda _{i}\left(\CA \right) ~ \leq ~ k + i + d
\end{equation}
where $d = \min \Big \{ (i+1)k, \ell \Big \}$.
\end{proposition}

Note that by construction $\ell \le n$ and in Overbeck's attack, the integer  $i$ is equal to $n-k - 1$ so that we have both 
$d = \ell$ and, with high probability, the upper bound in \eqref{ineq:dsg} is  actually an equality, namely 
\[
\dim \Lambda _{n-k-1}\left(\CA \right) = k + (n - k - 1) + d = n + \ell - 1.
\]
This implies that the dimension of $\dual{\Lambda _{i}\left(\CA \right)}$ is equal to $1$. This fact is 
then harnessed by \citep{O08} to recover an equivalent Gabidulin code which enables to decrypt any ciphertext. 

\begin{proposition}\citep{O08}
Assume that the public key is $\Gp = \Sm \left(\Xm \mid \Gm \right) \Pm$
with $\Xm \in \MS{k}{\ell}{\fqm}$, $\Pm \in \GL_{n+\ell}(\fq)$ and $\Gm$ generates a Gabdidulin code $\gab{k}{\gv}$.
If the dimension of $\Lambda_{n-k-1}\left(\CA \right)^\perp$ is equal to $1$, then it is possible to recover
(with $O\left( (n+\ell)^3 \right)$ field operations) alternative matrices $\Xm^* \in \MS{k}{\ell}{\fqm}$, $\Pm^* \in \GL_{n+\ell}(\fq)$ and $\Gm^*$ that generates a $(n,k)-$Gabdidulin code such that 
 \[
 \Gp = \Sm \left(\Xm^* \mid \Gm^* \right) \Pm^*
 \]
\end{proposition}

Overbeck's attack uses crucially  two important facts : 
 the column scrambler matrix $\Pm$ is defined on the based field $\fq$, 
 and the codimension of $\Lambda _{n-k-1}\left(\CA \right)$ is $1$.
Several works propose to resist to Overbeck's attack by taking special 
distortion matrix so that the second property is not true as \citep{L10,RGH10}. In the sequel, we focus on the variant of \cite{RGH10} which is usually called ``the smart approach''.  \cite{HMR17} also proposed an attack of that variant, focusing instead on recovering elements of rank one rather than the dimension of the dual space. In the next section, we present a new attack that exploits the structure of the dual space and is based on a new view of the ``smart'' approach \citep{RGH10}.

\section{A New Attack on the Smart Approach of the GPT Cryptosystem} \label{Our_Attack}

Here we describe the reparation given by \cite{RGH10} and we give our new algorithm that recovers an alternative secret key in polynomial time.
\subsection{Description}
The only difference is in the key generation algorithm, specially in the choice of $\Xm$. The authors proposed to take $\Xm \in \MS{k}{\ell}{\fqm}$ that is a concatenation of a $q-$Vandermonde matrix $\Xm_1 \in \MS{k}{a}{\fqm}$ and a random matrix $\Xm_2 \in \MS{k}{\ell-a}{\fqm}$ with $0 < a < \ell$. That is to say $\Xm = \left(\Xm_1 \mid \Xm_2 \right).$ 

With the above construction, it was shown by \cite{RGH10} that the corresponding public code $\Cp$ satisfies $\dim \Lambda_{n-k-1}\left(\Cp \right)^{\perp} \neq 1$ and hence, Overbeck's attack fails. In the following, we show that one can modify the public code $\Cp$ to make Overbeck's attack succeed again.  
\subsection{Cryptanalysis}
The attack we present here is inspired by the attack presented by \cite{OTK15} in a hamming metric context. In the sequel, $\Gm$ is a $q-$Vandermonde matrix from $\MS{k}{n}{\fqm}$ that generates a Gabidulin code $\gab{k}{\gv}$, $\gv$ being the first row of $\Gm$.  

Let $\Sm \in \GL_k(\fqm)$, $\Xm_2 \in \MS{k}{\ell-a}{\fqm}$, $\bv= \left(b_1, \cdots{}, b_a \right)$ and 
\begin{equation}
\Xm_1 =
\begin{pmatrix}
b_{1}^{[0]} & \cdots{} & b_{a}^{[0]} \\
\vdots{}      &              &  \vdots{} \\
b_{1}^{[k-1]} & \cdots{} & b_{a}^{[k-1]}
\end{pmatrix}.
\end{equation} 
We have $\Gp= \Sm \left(\Xm_1 \mid \Xm_2 \mid \Gm \right) \Pm$ with $\Pm \in \GL_{n+\ell}(\fq)$. 
We start the cryptanalysis by the following lemma:
\begin{lemma}\label{lem:attack}
There exists $\Pm^* \in \GL_{n+\ell}(\fq)$ and $\Gm^* \in \MS{k}{n+s}{\fqm}$ a generator matrix of a Gabidulin code  such that 
\[
\Gp= \Sm \left( \ZZ \mid \Xm_2 \mid \Gm^* \right) \Pm^*
\]
$s$ being an integer satisfying $0 \leq s \leq a$ and $n+s \leq m$.  
\end{lemma} 
\begin{proof}
Let $\gv^\prime = \left(\bv \mid \gv \right) \in \fqm^{a+n}$. Since $\norm{\gv^\prime} \geq \norm{\gv} =n$, let $s$ be an integer such that $\norm{\gv^\prime} = n+s$. Clearly, we have $s \leq a$ and $\norm{\left(\Xm_1 \mid \Gm \right)}=\norm{\gv^\prime}= n+s$ from the definition of the column rank together with Lemma \ref{lem:frob_prop}. From Proposition \ref{prop:rank_reduction},  there exists a matrix $\Qm \in \GL_{n+a}(\fq)$ such that $\left(\Xm_1 \mid \Gm \right) \Qm$ has $a-s$ many zero columns. More precisely,  $\left(\Xm_1 \mid \Gm \right) \Qm = \left( \ZZ \mid \Gm^* \right)$ where $\Gm^* \in \MS{k}{n+s}{\fqm}$ is a generator matrix of a Gabidulin code $\gab{k}{\gv^*}$ with $\gv^\prime \Qm = \left(\ZZ \mid \gv^* \right)$. This implies that there exists a matrix $\Rm \in \GL_{n+\ell}(\fq)$ such that $\left(\Xm_1 \mid \Xm_2 \mid \Gm \right) \Rm =\left( \ZZ \mid \Xm_2 \mid \Gm^* \right)$. To finish the proof, take $\Pm^*= \Rm^{-1}\Pm$. Since $\gv^\prime$ belongs to $\fqm^{a+n}$, we have $\norm{\gv^\prime} = n+s \leq m$ by the definition of the rank weight.     
\end{proof}
In the sequel, for any integer $\ell$ and given two linear codes $\CC$ and $\CC^{\prime}$, we will say that $\CC$ is the code $\CC^{\prime}$ with $\ell$ redundancies if a generator matrix of $\CC$ is a generator matrix of $\CC^{\prime}$ with $\ell$ additional columns.
Let $\Cp$ be the code generated by $\Gp$. We then have the following proposition:
\begin{proposition}\label{propo:redundancies}
The code $\Cp$ is the public code of a general GPT cryptosystem with $w=a-s$ redundancies. 
\end{proposition}
\begin{proof}
We have $\Gp= \Sm \left( \ZZ \mid \Xm_2 \mid \Gm^* \right) \Pm^*$. Let us suppose that $\Pm^* = 
\begin{pmatrix}
\Qm_1 \\
\Qm_2
\end{pmatrix}$
with $\Qm_1 \in \MS{w}{n+\ell}{\fq}$ and $\Qm_2 \in \MS{n+\ell-w}{n+\ell}{\fq}$. We have 
$\Gp = \Sm \left( \Xm_2 \mid \Gm^* \right) \Qm_2$ and $\rank{(\Qm_2)} = n+\ell-w$ \footnote{ $\rank{(\Qm_2)}$ is the rank over $\fqm$ in contrast to $\norm{\Qm_2}$, which denotes the column rank over $\fq$}. Let us suppose without loss of generality that the matrix $\Qm_2^*$ of the first $n+\ell-w$ columns of $\Qm_2$ is of full rank. Let $\Gp^*= \Sm \left( \Xm_2 \mid \Gm^* \right) \Qm_2^*$ and $\Xm = \Sm\left( \Xm_2 \mid \Gm^* \right) \Qm_2^{**}$ where $\Qm_2^{**}$ is the last $w$ columns of $\Qm_2$. Then $\Gp = \left( \Gp^* \mid \Xm \right)$. One can remark to finish that $\Gp^*$ is a generator matrix of a general GPT cryptosystem.      
\end{proof}

From the above proposition, if $w$ redundancies of $\Cp$ are identified and removed, a cryptanalyst can use Overbeck's attack to build an alternative secret key. One can remark that the notion of redundancies here is very related to the notion of ``information set'' \citep{P62} of the code $\CC_2$ generated by the matrix $\Qm_2$ in the proof of the previous proposition.  In fact, we are looking for an information set of $\CC_2$, since the corresponding columns of such a set in $\Gp$ will generate the public code of a general GPT cryptosystem. 
We will find such a set by eliminating progressively (one by one) $w$ columns of $\Gp$, corresponding to a ``redundancy set'' $\Im=\{ i_1, ..., i_w \} \subset \{ 1, 2, ...,n+\ell \}$ of $\CC_2$, or redundancies of $\Cp$. The reason we succeed by doing that is given by Proposition \ref{algoReason}. 
By abuse of language, a set of $w$ redundancies of $\Cp$ will be also called a "redundancy set" of $\Cp$. 

In the sequel, let $f=n+s-k$ and $\Lambda_f(\Cp)$ the code generated by $\Lambda_f(\Gp)$. For any subset $\Jm \subset \{ 1, 2, ...,n+\ell \}$, we denote by $\Cp^ {\Jm}$ the code obtained by puncturing $\Cp$ at all the positions corresponding to the elements in $\Jm$. For a matrix $\Mm$ with $n+\ell$ columns, $\Mm^{\Jm}$ is the matrix obtained from $\Mm$ by removing all the columns with positions corresponding to the  elements in $\Jm$. It is clear that $\Jm$ is a ``redundancy set'' of $\Cp$ if and only if $\Qm_2^{\Jm}$ is non singular. We have the following proposition:

\begin{proposition}\label{algoReason}
A set $\Im=\{ i_1, ..., i_w \} \subset \{ 1, 2, ...,n+\ell \}$ is a ``redundancy set'' of $\Cp$ if and only if for any subset $\Jm \subset \Im$,  
\[
\dim{\Lambda_f(\Cp^ {\Jm})}= n+s+ \ell -a
\]

\end{proposition}
\begin{proof}
We have:
\[
\dim{\Lambda_f(\Cp^{\Jm})} = \rank{({\Lambda_f(\Gp^{\Jm})})}=\rank{(\Lambda_f\left(\Xm_2 \mid \Gm^*\right)\Qm^{\Jm}_2)}
\]
 Since $\Xm_2$ is a random matrix, with a high probability we have
 \[
\dim{\Lambda_f(\Cp^{\Jm})} = \min\{\rank{(\Lambda_f\left(\Xm_2 \mid \Gm^*\right))}, \rank{(\Qm^{\Jm}_2)} \} = \min\{n+s+\ell-a, \rank{(\Qm^{\Jm}_2)} \}
\]
If $\Im$ is a ``redundancy set'' of $\Cp,$ we will always have $\rank{(\Qm^{\Jm}_2)} = \rank{(\Qm_2)}=n+s+\ell-a$ and $\dim{\Lambda_f(\Cp^{\Jm})}=n+s+\ell-a$ for all $\Jm \subset \Im$. Else, we will have $\rank{(\Qm^{\Im}_2)} < n+s+\ell-a$ and $\dim{\Lambda_f(\Cp^{\Im})} < n+s+\ell-a$. 
\end{proof}

It is easy for an adversary to use the previous proposition to identify a ``redundancy set'' $\Im$.  To fully break the system, one can apply Overbeck's attack with $f= n+s-k-1$, but the value of $s$ is not known. For the case $m=n$, it is easy to see thanks to Lemma \ref{lem:attack}  that $s$ is equal to $0$ and in a general context $(n \leq m)$, one can remark from the same lemma that the integer $s$ is the smallest one that satisfies 
\[
\rank{(\Lambda_{n+s-k}(\Gp))}=\rank{(\Lambda_{n+s+1-k}(\Gp))}
\] 
We summarise the attack in Algorithm \ref{attack}.     
\begin{algorithm}[ht]
    \caption{Key Recovery of The Smart Approach of the GPT Cryptosystem } 
 \hspace*{\algorithmicindent} \textbf{Input} : $\Gp, t, a, k, n, \ell$ \\
 \hspace*{\algorithmicindent} \textbf{Output} : $s, \Jm=\{ i_1, ..., i_w \}, \Pm^* \in \GL_{n+\ell-a+s}(\fq)$ 
    \begin{algorithmic}[1]
    	\State{$ s \gets a$}      
        \While{$\rank{(\Lambda_{n+s-k}(\Gp))}=\rank{(\Lambda_{n+s+1-k}(\Gp))}$}
            \State{$s \gets s-1$}
        \EndWhile
        \State{$s \gets s+1$}
        \State{$w \gets a-s$}
        \State{$y \gets n+s+\ell-a $}
        \State{$f \gets n+s-k$}
        \State{$Z \gets \{ 1,...,Length(\Cp) \}$ and $J \gets [~ ]$  }
        \State{$ j \gets Random(Z)$}
        \While{$\sharp J \neq w$}
        	\If{$\dim{(\Lambda_{n+s-k}(\Cp^j))} = y $}
        		\State{$J \gets HorizontalJoin(J , [j]) $}
        		\State{$\Cp \gets \Cp^j$}
        		\State{$Z \gets \{ 1,...,Length(\Cp) \}$}
        		\State{$ j \gets Random(Z)$}
        	\Else
        		\State{$Z \gets Z \setminus \{ j \}$}
        		\State{$ j \gets Random(Z)$}
        	\EndIf
        \EndWhile
        \State{\Return $\Cp$, $J$}
        \State{Define $\Gp$ as the generator matrix of $\Cp$}
		\State{Apply Overbeck's algorithm on $\Gp$ with $f=n+s-k-1$ }
	\end{algorithmic}
	\label{attack}        
\end{algorithm}

\subsection{Complexity and Experimental Results}

During the computation phase of $s$, the main computations are 
$\rank{(\Lambda_{n+s-k}(\Gp))}$ and \\
$\rank{(\Lambda_{n+s+1-k}(\Gp))}$ which are computed with a complexity $O(k^3(n+a+1-k)^3)$. To identify a set of $w=a-s$ random redundancies, the main computation is $\dim{(\Lambda_{n+s-k}(\Cp^j))}$ (for $j \in \{1, ...,n+ \ell\}$) which can be done only once. So the complexity of this step is $O(k^3(n+s+1-k)^3)$. By considering the final step that consists to apply Overbeck's attack, the overall complexity is $O(k^3(n+a+1-k)^3)$ operations on $\fqm$ since the complexity of this final step is $O((n+\ell)^3)$ operations on $\fqm$. We implemented the attack (for $m \leq 30$ and for several values of $a$ such that $am \geq 60$ as proposed by \cite{RGH10}) on Magma V2.21-6 and a secret key was always found in less than $5$ seconds. This  confirms the efficiency of the approach.   
\\

\subsection{A Brief Comparison with the Attack from {\small \cite{HMR17}} }
Let $\Gp$ be a public generator matrix of the smart approach of the GPT cryptosystem. We recall that $\Gp= \Sm \left(\Gm \mid \Xm_1 \mid \Xm_2 \right) \Pm$ with $\Pm \in \GL_{n+\ell}(\fq)$. \cite{HMR17} decomposed $\Gp$ as follows:   
\[
\Gp= \Sm \left(\Gm \mid \Xm_1 \mid \ZZ \right) \Pm + \Sm \left(\ZZ \mid \ZZ \mid \Xm_2 \right) \Pm = \Sm \left(\Gm \mid \Xm_1 \mid \ZZ \right) \Pm + \Sm  \Xm_2 \Pm^{\prime}
\]
where $\Pm^{\prime} \in \MS{t-a}{n+\ell}{\fq}$ is a full rank matrix. They then consider the $\fqm-$linear code $\CC^\prime$ generated by $\Pm^{\prime}$ and show that the only elements of rank one in $\Lambda_{t-a}(\Cp)$ are those from $\CC^\prime,$ after showing that $\CC^\prime \subset \Lambda_{t-a}(\Cp).$ This means that one can reconstruct the matrix $\Pm^{\prime}$ or an alternative $\Pm^{*}$ by collecting elements of rank one in $\Lambda_{t-a}(\Cp)$  (these elements are found by solving a linear system parametrised by $\Lambda_{t-a}(\Gp)$). They finish their attack by computing a parity check matrix $\Hm_{\Pm^*} \in \MS{n + \ell - t+a }{n + \ell}{\fq}$ of $\CC^\prime$ from $\Pm^{*}$. Note that $\Hm_{\Pm^*}^T$ can be use in place of $\Pm^{-1}$ during the decoding process, since $\Sm  \Xm_2 \Pm^{\prime} \Hm_{\Pm^*}^T = \ZZ$ and $\Gp \Hm_{\Pm^*}^T= \Sm \left(\Gm \mid \Xm_1 \mid \ZZ \right) \Pm \Hm_{\Pm^*}^t$ is a $q-$Vandermonde matrix. This attack have a complexity in $O(k^2 n m^2 (t^2 + k))$ operations in $\fq$. 

Our view of $\Gp$ is different. As explained in Proposition \ref{propo:redundancies}, $\Gp$ is a public matrix of a general GPT cryptosystem concatenated with some useless columns. We use the Frobenius operator to extract (from $\Gp$) a public matrix $\Gp^*$ of the general GPT cryptosystem, and the last step is to apply Overbeck's attack to $\Gp^*$. We end-up with $O(k^3(n+a+1-k)^3)$ operations on $\fqm$. 

One can remark that the cost of \cite{HMR17} is better if we assume that $m$ and $n$ are equal or almost equal. But for very large $m$ compare to $n$, the cost of our attack is better.  
\\            
\section{Conclusion}
In this paper, we have shown that the smart approach of the GPT cryptosystem proposed by \cite{RGH10} to avoid Overbeck's structural attack, is equivalent to insert some redundancies to the public code of a standard GPT cryptosystem. We show how to remove the redundancies in order to be able to apply Overbeck's attack on the public code obtained. This allows a cryptanalyst to built an alternative secret key in polynomial time. Our attack is a new and simple way to show that the smart approach of the GPT cryptosystem is not secure.    
%

\section{Acknowledgement}
The author would like to thank the anonymous reviewers for their valuable comments and suggestions which helped to improve the paper.

\bibliographystyle{chicago}
\bibliography{codecrypto2}
\Addresses
\end{document}